\documentclass{article}
\usepackage{fullpage}
\usepackage{amsfonts,amssymb,amsmath}
\usepackage{graphicx}
\usepackage{color}

\def\jump{\vskip0.05in}

\def\E{{\mathbb E}}

\def\R{{\mathbb R}}

\def\pr{{\mbox{\rm Pr}}}

\def\qed{\vrule height8pt width3pt depth0pt}

\def\qed{\vrule height8pt width3pt depth0pt}

\def\cost{{\mbox{\rm cost}}}
\def\naesat{{\mbox{$\mbox{NAESAT}^*$}}}
\def\leaves{{\mbox{\tt leaves}}}
\newtheorem{thm}{Theorem}
\newtheorem{lemma}[thm]{Lemma}

\newtheorem{cor}[thm]{Corollary}

\newenvironment{proof}{\noindent{\it Proof.} }{\qed\jump}

\title{A cost function for similarity-based hierarchical clustering}
\author{Sanjoy Dasgupta}

\begin{document}

\maketitle

\begin{abstract}
The development of algorithms for hierarchical clustering has been hampered by a shortage of precise objective functions. To help address this situation, we introduce a simple cost function on hierarchies over a set of points, given pairwise similarities between those points. We show that this criterion behaves sensibly in canonical instances and that it admits a top-down construction procedure with a provably good approximation ratio.
\end{abstract}

\section{Introduction}

A {\it hierarchical clustering} is a recursive partitioning of a data set into successively smaller clusters. It is represented by a rooted tree whose leaves correspond to the data points, and each of whose internal nodes represents the cluster of its descendant leaves.

A hierarchy of this sort has several advantages over a {\it flat clustering}, which is a partition of the data into a fixed number of clusters. First, there is no need to specify the number of clusters in advance. Second, the output captures cluster structure at all levels of granularity, simultaneously.

There are several well-established methods for hierarchical clustering, the most prominent among which are probably the bottom-up agglomerative methods: single linkage, average linkage, and complete linkage (see, for instance, Chapter 14 of \cite{HTF09}). These are widely used and are part of standard packages for data analysis. Despite this, there remains an aura of mystery about the kinds of clusters that they find. In part, this is because they are specified procedurally rather than in terms of the objective functions they are trying to optimize. For many hierarchical clustering algorithms, it is hard to imagine what the objective function might be.

This is unfortunate, because the use of objective functions has greatly streamlined the development of other kinds of data analysis, such as classification methods. Once a cost function is specified, the problem definition becomes precise, and it becomes possible to study computational complexity and to compare the efficacy of different proposed algorithms. Also, in practice prior information and other requirements often need to be accommodated, and there is now a well-weathered machinery for incorporating these within cost functions, as explicit constraints or as regularization terms.

In this work we introduce a simple cost function for hierarchical clustering: a function that, given pairwise similarities between data points, assigns a score to any possible tree on those points. Through a series of lemmas (Section~\ref{sec:basics}), we build up intuitions about the kinds of hierarchies that this function favors. Its behavior on some canonical examples---points on a line, complete graphs, and planted partitions---can be characterized readily (Section~\ref{sec:examples}), and corresponds to what one would intuitively want in these cases. The cost function turns out to be NP-hard to optimize (Section~\ref{sec:hardness}), a fate it shares with all the common cost functions for flat clustering. However, it has a provably-good approximation algorithm (Section~\ref{sec:alg}): a simple top-down heuristic, variants of which are used widely in graph partitioning.

\subsection*{Related work}

Much of the development of hierarchical clustering has taken place within the context of phylogenetics~\cite{SS63,JS71,F04}. Several of the methods originating in this field, such as average linkage, have subsequently been adapted for more general-purpose data analysis. This literature has also seen a progression of cost functions over taxonomies. One of the earliest of these is {\it parsimony}: given a collection of vectors $x_1, \ldots, x_n \in \{0,1\}^p$, the idea is to find a tree whose leaves correspond to the $x_i$'s and whose internal nodes are also marked with vectors from the same space, so that total change along the branches of the tree (measured by Hamming distance, say) is minimized. Later work has moved towards probabilistic modeling, in which a tree defines a probability space over observable data and the goal is to find the maximum likelihood tree. This is similar in flavor to parsimony, but facilitates the addition of latent variables such as differential rates of mutation along different branches. Finally, there is the fully Bayesian setting in which a prior is placed over all possible trees and the object is to sample from the posterior distribution given the observations. Although some of the probabilistic models have been used for more general data (for instance, \cite{N03}), they tend to be fairly complex and often implicitly embody constraints---such as a ``molecular clock'', in which all leaves are expected to be roughly equidistant from the root---that are questionable outside the biological context. A more basic and interesting question is to find variants of parsimony that behave reasonably in general settings and that admit good algorithms. This is not the agenda of the present paper, however. For instance, our cost function does not assume that the data lie in any particular space, but is based purely on pairwise similarities between points.

The literature on hierarchical clustering also spans many other disciplines, and there are at least two other prominent strains of work addressing some of the same concerns that have motivated our search for a cost function. 

One of these has to do with the statistical consistency of hierarchical clustering methods, as pioneered by Hartigan~\cite{H85}. The aim here is to establish that if data is sampled from a fixed underlying distribution, then the tree structure obtained from the data converges in some suitable sense as the sample size grows to infinity. Only a few methods have been shown to be consistent---single linkage in one dimension, and some variants of it in higher dimension~\cite{CDKL14,EBW15}---and it is an open question to assess the convergence properties, if any, of other common methods such as average linkage. 

Finally, there has been an effort to evaluate hierarchies in terms of cost functions for flat clustering---by evaluating the $k$-means cost, for example, of a hierarchy truncated to $k$ clusters. Results of this kind are available for several standard clustering objectives, including $k$-center, $k$-median, and $k$-means~\cite{DL05,P06,LNRW10}. In the present paper, we take a different approach, treating the hierarchies themselves as first-class objects to be optimized.

On the algorithmic front, the method we ultimately present (Section~\ref{sec:alg}) uses a recursive splitting strategy that is standard in multiway graph partitioning~\cite{KL70}. Our results can be seen as providing some justification for this heuristic, by presenting a concrete optimization problem that it approximately solves.

\section{The cost function and its basic properties}
\label{sec:basics}

We assume that the input to the clustering procedure consists of pairwise similarities between $n$ data points. Although it would perhaps be most natural to receive this information in matrix form, we will imagine that it is presented as a weighted graph, as this will facilitate the discussion. The input, then, is an undirected graph $G = (V,E,w)$, with one node for each point, edges between pairs of similar points, and positive edge weights $w(e)$ that capture the degree of similarity. We will sometimes omit $w$, in which case all edges are taken to have unit weight.

We would like to hierarchically cluster the $n$ points in a way that is mindful of the given similarity structure. To this end, we define a cost function on possible hierarchies. This requires a little terminology. Let $T$ be any rooted, not necessarily binary, tree whose leaves are in one-to-one correspondence with $V$. For any node $u$ of $T$, let $T[u]$ be the subtree rooted at $u$, and let $\leaves(T[u]) \subseteq V$ denote the leaves of this subtree. For leaves $i,j \in V$, the expression $i \vee j$ denotes their lowest common ancestor in $T$. Equivalently, $T[i \vee j]$ is the smallest subtree whose leaves include both $i$ and $j$ (Figure~\ref{fig:some-defns}).

\begin{figure}
\begin{center}
\includegraphics[width=2in]{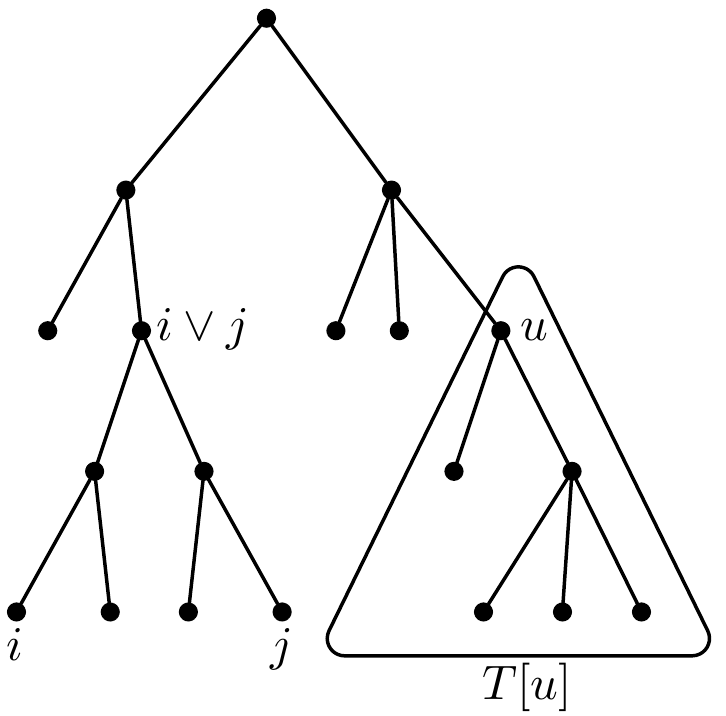}
\end{center}
\caption{Induced subtrees and least common ancestors.}
\label{fig:some-defns}
\end{figure}

The edges $\{i,j\}$ in $G$, and their strengths $w_{ij}$, reflect locality. When clustering, we would like to avoid cutting too many edges. But in a hierarchical clustering, all edges do eventually get cut. All we can ask, therefore, is that edges be cut {\it as far down the tree as possible}. Accordingly, we define the cost of $T$ to be
$$ \mbox{cost}_G(T) = \sum_{\{i,j\} \in E} w_{ij} \, |\leaves(T[i \vee j])| .$$
(Often, we will omit the subscript $G$.) If an edge $\{i,j\}$ of unit weight is cut all the way at the top of the tree, it incurs a penalty of $n$. If it is cut further down, in a subtree that contains $\alpha$ fraction of the data, then it incurs a smaller penalty of $\alpha n$.

\begin{figure}
\begin{center}
\includegraphics[width=4in]{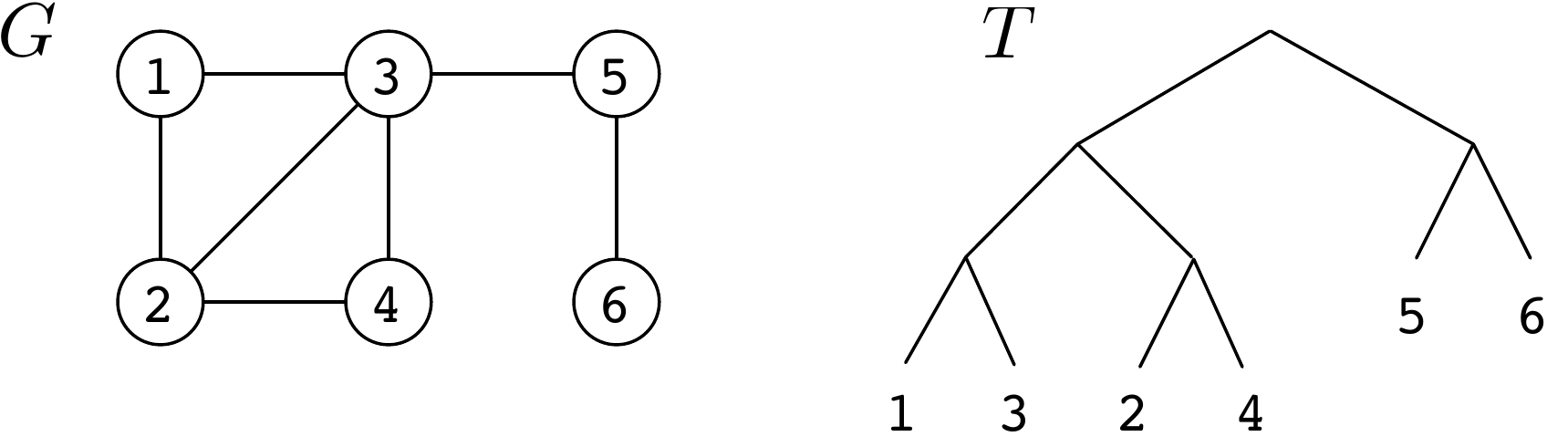}
\end{center}
\caption{A small graph $G$ and a candidate hierarchical clustering.}
\label{fig:toy-example}
\end{figure}

Figure~\ref{fig:toy-example} shows a toy example. Edge $\{3,5\}$ is cut at the root of $T$ and incurs a cost of 6. Edges $\{1,2\}, \{2,3\}, \{3,4\}$ each incur a cost of 4, and the remaining three edges cost 1 apiece. Thus $\cost_G(T) = 6 + 3 \times 4 + 3 \times 1 = 21$.

\subsection{Two interpretations of the cost function}

A natural interpretation of the cost function is in terms of cuts. Each internal node of the tree corresponds to a {\it split} in which some subset of nodes $S \subseteq V$ is partitioned into two or more pieces. For a binary split $S \rightarrow (S_1, S_2)$, the splitting cost is $|S| \, w(S_1, S_2)$, where $w(S_1, S_2)$ is the weight of the cut,
$$
w(S_1, S_2) \ = \ \sum_{\{i,j\} \in E: \ i \in S_1, j \in S_2} w_{ij} .
$$
In the example of Figure~\ref{fig:toy-example}, the top split, $\{1,2,3,4,5,6\} \rightarrow (\{1,2,3,4\}, \{5,6\})$ costs 6, while the subsequent split $\{1,2,3,4\} \rightarrow (\{1,3\}, \{2,4\})$ costs $12$.

This extends in the obvious way to $k$-ary splits $S \rightarrow (S_1, S_2, \ldots, S_k)$, whose cost is $|S| w(S_1, \ldots, S_k)$, with
$$w(S_1, \ldots, S_k) \ = \  \sum_{1 \leq i < j \leq k} w(S_i, S_j) .$$
The cost of a tree $T$ is then the sum, over all internal nodes, of the splitting costs at those nodes:
\begin{equation}
\cost(T) = \sum_{\mbox{\small splits $S \rightarrow (S_1, \ldots, S_k)$ in $T$}} |S| w(S_1, \ldots, S_k) . \label{eq:split-formulation}
\end{equation}
We would like to find the hierarchy $T$ that minimizes this cost.

Alternatively, we can interpret a tree $T$ as defining a distance function on points $V$ as follows:
$$ d_T(i,j) = |\leaves(T[i \vee j])| - 1.$$
This is an ultrametric, that is, $d_T(i,j) \leq \max(d_T(i,k), d_T(j,k))$ for all $i,j,k \in V$. The cost function can then be expressed as
$$ \mbox{cost}(T) = \sum_{\{i,j\} \in E} w_{ij} d_T(i,j) \ \ + \mbox{constant}.$$
Thus, we seek a tree that minimizes the average distance between similar points.

\subsection{Modularity of cost}

An operation we will frequently perform in the course of analysis is to replace a subtree of a hierarchical clustering $T$ by some other subtree. The overall change in cost has a simple expression that follows immediately from the sum-of-splits formulation of equation~(\ref{eq:split-formulation}).
\begin{lemma}
Within any tree $T$, pick an internal node $u$. Suppose we obtain a different tree $T'$ by replacing $T[u]$ by some other subtree $T_u'$ with the same set of leaves. Then
$$ \cost(T') = \cost(T) - \cost(T[u]) + \cost(T_u'),$$
where the costs of $T[u]$ and $T_u'$ are with respect to the original graph restricted to $\leaves(T[u])$.
\label{lemma:modularity}
\end{lemma}

\subsection{The optimal tree is binary}

In the figure below, the tree on the left is at least as costly as that on the right.

\begin{center}
\raisebox{.25in}{\includegraphics[width=2.5in]{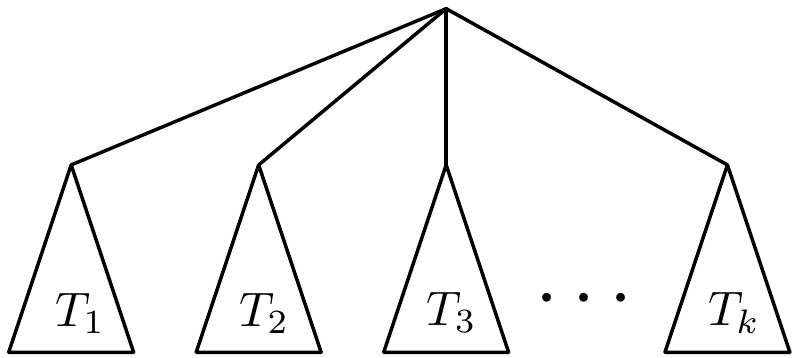}}
\hskip.75in
\includegraphics[width=2in]{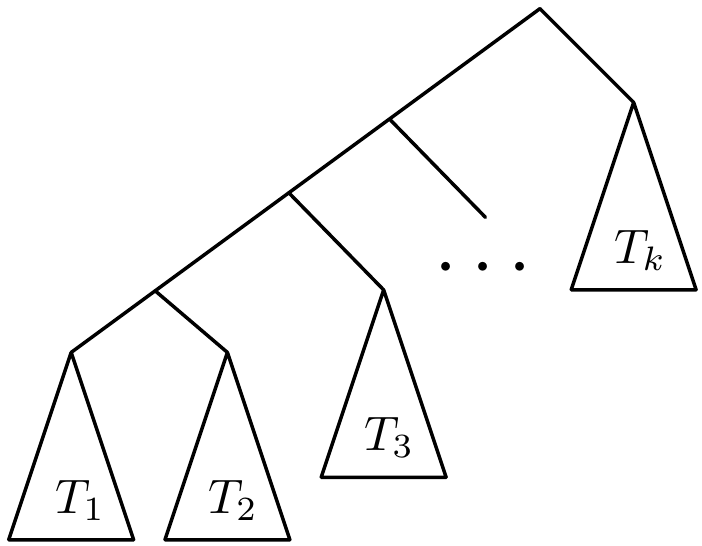}
\end{center}  

\noindent
It follows that there must always exist an optimal tree that is binary.

\subsection{Different connected components must first be split apart}

If $G$ has several connected components, we would expect a hierarchical clustering to begin by pulling these apart. This natural property holds under the cost function we are considering.
\begin{lemma}
Suppose an optimal tree $T$ contains a subtree $T'$ whose leaves induce a subgraph of $G$ that is not connected. Let the topmost split in $T'$ be $\leaves(T') \rightarrow (V_1, V_2)$. Then $w(V_1, V_2) = 0$.
\label{lemma:connected-components}
\end{lemma}
\begin{proof}
Suppose, by way of contradiction, that $w(V_1, V_2) > 0$. We'll construct a subtree $T''$ with the same leaves as $T'$ but with lower cost---implying, by modularity of the cost function (Lemma~\ref{lemma:modularity}), that $T$ cannot be optimal.

For convenience, write $V_o = \leaves(T')$, and let $U_1, U_2$ be a partition of $V_o$ into two sets that are disconnected from each other (in the subgraph of $G$ induced by $V_o$). $T''$ starts by splitting $V_o$ into these two sets, and then copies $T'$ on each side of the split.

\begin{center}
\includegraphics[width=2in]{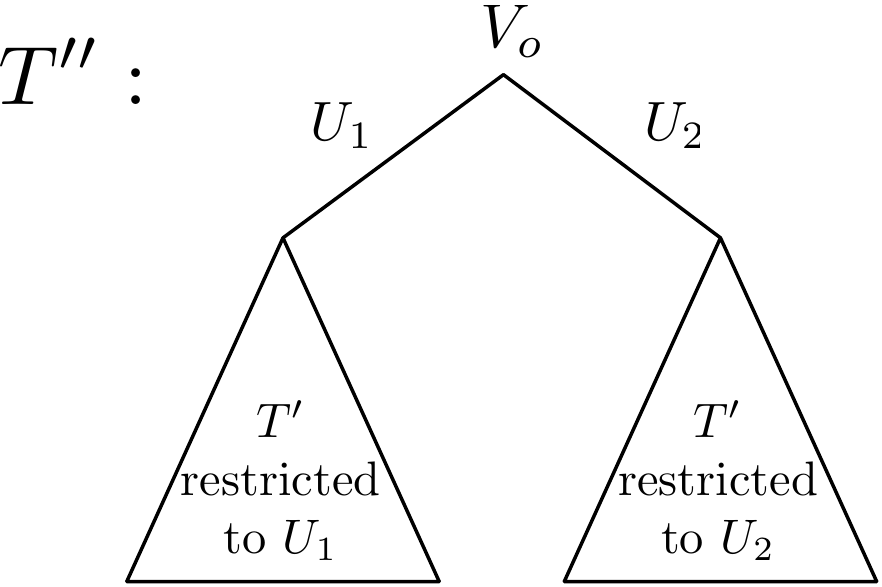}
\end{center}

The top level of $T''$ has cost zero. Let's compare level $\ell + 1$ of $T''$ to level $\ell$ of $T'$. Suppose this level of $T'$ contains some split $S \rightarrow (S_1, S_2)$. The cost of this split---its contribution to the cost of $T'$---is
$$ \mbox{cost of split in $T'$} = |S| w(S_1, S_2) .$$
Now, $T''$ has two corresponding splits, $S \cap U_1 \rightarrow (S_1 \cap U_1, S_2 \cap U_1)$ and $S \cap U_2 \rightarrow (S_1 \cap U_2, S_2 \cap U_2)$. Their combined cost is
\begin{align*}
\mbox{cost of splits in $T''$}
&= |S \cap U_1| w(S_1 \cap U_1, S_2 \cap U_1) + |S \cap U_2| w(S_1 \cap U_2, S_2 \cap U_2) \\
&\leq |S| w(S_1 \cap U_1, S_2 \cap U_1) + |S| w(S_1 \cap U_2, S_2 \cap U_2) \\
&= |S| w(S_1, S_2) \ = \ \mbox{cost of split in $T'$},
\end{align*}
with strict inequality if $S \cap U_1$ and $S \cap U_2$ are each nonempty and $w(S_1, S_2) > 0$. The latter conditions hold for the very first split of $T'$, namely $V_o \rightarrow (V_1, V_2)$.

Thus, every level of $T'$ is at least as expensive as the corresponding level of $T''$ and the first level is strictly more expensive, implying $\mbox{cost}(T'') < \mbox{cost}(T')$, as claimed.
\end{proof}

\section{Illustrative examples}
\label{sec:examples}

To get a little more insight into our new cost function, let's see what it does in three simple and canonical situations: the line graph, the complete graph, and random graphs with planted partitions.

\subsection{The line}

Consider a line on $n$ nodes, in which every edge has unit weight:

\vskip.15in
\begin{center}
\includegraphics[width=4in]{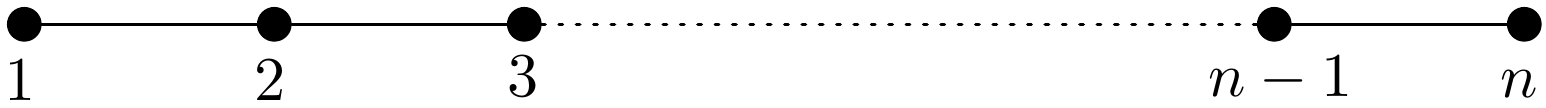}
\end{center}

\vskip.1in

A hierarchical clustering need only employ splits consisting of a single edge. This is because any split that involves removing multiple edges $\{e_1, e_2, \ldots\}$ can be replaced by a series of single-edge splits (first $e_1$, then $e_2$, and so on), with an accompanying reduction in cost.

Let $C(n)$ denote the cost of the best tree for the line on $n$ nodes. The removal of an edge incurs a cost of $n$ and creates two smaller versions of the same problem. Thus, for $n > 1$,
$$ C(n) = n + \min_{1 \leq j \leq n-1} C(j) + C(n-j) .$$
The best strategy is to split as evenly as possible, into sets of size $\lfloor n/2 \rfloor$ and $\lceil n/2 \rceil$, yielding an overall cost $C(n) = n \log_2 n + O(n)$.

On the other hand, the worst splitting strategy is to choose an extremal edge, which when applied successively results in a total cost of
$$ n + (n-1) + (n-2) + \cdots + 2 = \frac{n(n+1)}{2} - 1 .$$

\subsection{The complete graph}

What is the best hierarchical clustering for the complete graph on $n$ nodes, with unit edge weights? It turns out that all trees have exactly the same cost, which is intuitively pleasing---and will be a crucial tool in our later analysis.

\begin{thm}
Let $G$ be the complete graph on $n$ nodes, with unit edge weights. For any tree $T$ with $n$ leaves,
$$ \cost_G(T) = \frac{1}{3}(n^3 - n).$$
\label{thm:complete}
\end{thm}
\begin{proof}
For a clique on $n$ nodes, a crude upper bound on the best achievable cost is $n \cdot {n \choose 2}$, assessing a charge of $n$ for every edge. We can round this up even further, to $n^3$.

Pick any tree $T$, and uncover it one split at a time, starting at the top. At each stage, we will keep track of the actual cumulative cost of the splits seen so far, along with a crude upper bound $\Phi$ on the cost yet-to-be-incurred. If, for instance, we have uncovered the tree to the point where the bottom nodes correspond to sets of size $n_1, n_2, \ldots$ (summing to $n$), then the bound on the cost-to-come is $\Phi = n_1^3 + n_2^3 + \cdots$. 

Initially, we are at the root, have paid nothing yet, and our crude bound on the cost of what lies below is $\Phi_0 = n^3$. Let's say the $t^{\rm th}$ split in the tree, $1 \leq t \leq n-1$, breaks a set of size $m$ into sets of size $k$ and $m-k$. For this split, the immediate cost incurred is $c_t = mk(m-k)$ and the crude upper bound on the cost yet-to-be-incurred shrinks by 
$$ \Phi_{t-1} - \Phi_t = m^3 - (k^3 + (m-k)^3) = 3mk(m-k) = 3c_t .$$

When we finally reach the leaves, the remaining cost is zero, but our crude upper bound on it is $\Phi_{n-1} = n$. Thus the total cost incurred is
$$ \cost(T) = c_1 + c_2 + \cdots + c_{n-1} = \frac{\Phi_0 - \Phi_1}{3} + \frac{\Phi_1 - \Phi_2}{3} + \cdots +  \frac{\Phi_{n-2} - \Phi_{n-1}}{3} = \frac{\Phi_0 - \Phi_{n-1}}{3} = \frac{n^3 - n}{3},$$
as claimed.
\end{proof}

\subsection{Planted partitions}

The behavior of clustering algorithms on general inputs is usually quite difficult to characterize. As a result, a common form of analysis is to focus upon instances with ``planted'' clusters.

The simplest such model, which we shall call the {\it $(n,p,q)$-planted partition} (for an even integer $n$ and $0 \leq q < p \leq 1$), assumes that there are two clusters, each containing $n/2$ points, and that the neighborhood graph $G = (V = [n], E)$ is constructed according to a generative process. For any pair of distinct nodes $i,j$, an edge is placed between them with probability
$$
\begin{array}{ll}
p & \mbox{if $i,j$ are in the same cluster} \\
q & \mbox{otherwise}
\end{array}
$$
and these ${n \choose 2}$ decisions are made independently. All edges in $E$ are then assigned unit weight. It has been found, for instance, that spectral methods can recover such partitions for suitable settings of $p$ and $q$~\cite{B85a,M01}.

We will also consider what we call the {\it general planted partition model}, which allows for an arbitrary number of clusters $C_1, C_2, \ldots, C_k$, of possibly different sizes, and has a less constrained generative process. For any pair $i,j$, the probability that there is an edge between $i$ and $j$ is
$$
\begin{array}{ll}
>q & \mbox{if $i,j$ are in the same cluster} \\
q & \mbox{otherwise}
\end{array}
$$
and these decisions need not be independent. Once again, the edges have unit weight.

\subsubsection{The general planted partition model}

Fix any set of points $V = [n]$, with unknown underlying clusters $C_1, \ldots, C_k$. Because the neighborhood graph $G$ is stochastic, we will look at the expected cost function,
\begin{align*}
\E [\cost_G(T)] 
&= \sum_{\{i,j\}} \E_G[w_{ij}] \, |\leaves(T[i \vee j])| \\
&= \sum_{\{i,j\}} \pr(\mbox{edge between $i$ and $j$ in $G$}) \, |\leaves(T[{i \vee j}])|.
\end{align*}

\begin{lemma}
For a general planted partition model with clusters $C_1, \ldots, C_k$, define graph $H = (V, F, w)$ to have edges between points in the same cluster, with weights
$$ w(i,j) = \pr(\mbox{edge between $i$ and $j$}) - q  > 0 .$$
In particular, $H$ has $k$ connected components corresponding to the clusters.

Now fix any tree $T$ with leaves $V$. For $G$ generated stochastically according to the partition model,
$$ \E [\cost_G(T)] = \frac{1}{3} q (n^3 - n) + \cost_H(T) .$$
\label{lemma:general-partition}
\end{lemma}
\begin{proof}
Letting $K_n$ denote the complete graph on $n$ nodes (with unit edge weights), we have
\begin{align*}
\E [\cost_G(T)]
&= 
\sum_{\{i,j\}} \pr(\mbox{edge between $i$ and $j$ in $G$}) \, |\leaves(T[{i \vee j}])| \\
&=
q \sum_{\{i,j\}} |\leaves(T[i \vee j])| + \sum_{\{i,j\}} (\pr(\mbox{edge between $i$ and $j$ in $G$}) -q) \, |\leaves(T[{i \vee j}])| \\
&=
q \sum_{\{i,j\}} |\leaves(T[i \vee j])| + \sum_{\{i,j\} \in F} w(i,j) \, |\leaves(T[{i \vee j}])| \\
&=
q \, \cost_{K_n}(T) + \cost_H(T).
\end{align*}
We then invoke Theorem~\ref{thm:complete}.
\end{proof}

Recall that the graph $H$ of Lemma~\ref{lemma:general-partition} has exactly $k$ connected components corresponding to the underlying clusters $C_1, \ldots, C_k$. By Lemma~\ref{lemma:connected-components}, any tree that minimizes $\cost_H(T)$, and thus $\E [\cost_G(T)]$, must first pull these components apart.
\begin{cor}
Pick any tree $T$ with leaves $V$. Under the general planted partition model, $T$ is a minimizer of $\E [\cost_G(\cdot)]$ if and only if, for each $1 \leq i \leq k$, $T$ cuts all edges between cluster $C_i$ and $V \setminus C_i$ before cutting any edges within $C_i$.
\end{cor}

Thus, in expectation, the cost function behaves in a desirable manner for data from a planted partition. It is also interesting to consider what happens when sampling error is taken into account. This requires somewhat more intricate calculations, for which we restrict attention to simple planted partitions.

\subsubsection{High probability bounds for the simple planted partition model}

Suppose graph $G = (V = [n], E)$ is drawn according to an $(n,p,q)$-planted partition model, as defined at the beginning of this section. We have seen that a tree $T$ with minimum {\it expected} cost $\E [ \cost_G(T)]$ is one that begins with a bipartition into the two planted clusters (call them $L$ and $R$).

Moving from expectation to reality, we hope that the optimal tree for $G$ will contain an almost-perfect split into $L$ and $R$, say with at most an $\epsilon$ fraction of the points misplaced. To this end, for any $\epsilon > 0$, define a tree $T$ (with leaves $V$) to be {\it $\epsilon$-good} if it contains a split $S \rightarrow (S_1, S_2)$ satisfying the following two properties.
\begin{itemize}
\item $|S \cap L| \geq (1-\epsilon) n/2$ and $|S \cap R| \geq (1-\epsilon)n/2$. That is, $S$ contains almost all of $L$ and of $R$.
\item $|S_1 \cap L| \leq \epsilon n/2$ and $|S_2 \cap R| \leq \epsilon n/2$ (or with $S_1, S_2$ switched). That is, $S_1 \approx R$ and $S_2 \approx L$, or the other way round.
\end{itemize}
As we will see, a tree that is not $\epsilon$-good has expected cost quite a bit larger than that of optimal trees and, for $\epsilon \approx ((\log n)/n)^{1/2}$, this difference overwhelms fluctuations due to variance in $G$.

To begin with, recall from Lemma~\ref{lemma:general-partition} that when $G$ is drawn from the planted partition model, the expected cost of a tree $T$ is
\begin{equation}
\E [ \cost_G(T) ] \ = \ \frac{1}{3} q (n^3 - n) + (p-q) \cost_H(T),
\label{eq:expected-cost}
\end{equation}
where $H$ is a graph consisting of two disjoint cliques of size $n/2$ each, corresponding to $L$ and $R$, with edges of unit weight.

We will also need to consider subgraphs of $H$. Therefore, define $H(\ell,r)$ to be a graph with $\ell + r$ nodes consisting of two disjoint cliques of size $\ell$ and $r$. We have already seen (Lemma~\ref{lemma:connected-components} and Theorem~\ref{thm:complete}) that the optimal cost for such a graph is
$$ C(\ell,r) \ \stackrel{\rm def}{=} \ \frac{1}{3}(\ell^3-\ell) + \frac{1}{3} (r^3 - r),$$
and is achieved when the top split is into $H(\ell,0)$ and $H(0,r)$. Such a split can be described as {\it laminar} because it conforms to the natural clusters. We will also consider splits of $H(\ell,0)$ or of $H(0,r)$ to be laminar, since these do not cut across different clusters. As we shall now see, any non-laminar split results in increased cost.
\begin{lemma}
Consider any hierarchical clustering $T$ of $H(\ell,r)$ whose top split is into $H(\ell_1,r_1)$, $H(\ell_2,r_2)$. Then
$$ \cost_{H(\ell,r)}(T) \geq C(\ell,r) + \ell_1\ell_2 r + r_1 r_2 \ell .$$
\label{lemma:two-cliques}
\end{lemma}
\begin{proof}
The top split cuts $\ell_1\ell_2 + r_1r_2$ edges. The cost of the two resulting subtrees is at least $C(\ell_1,r_1) + C(\ell_2,r_2)$. Adding these together,
\begin{align*}
\cost(T)
&\geq
(\ell_1 \ell_2 + r_1r_2)(\ell+r) + C(\ell_1,r_1) + C(\ell_2,r_2) \\
&=
(\ell_1 \ell_2 + r_1r_2)(\ell+r) + \frac{1}{3} \left( \ell_1^3 + \ell_2^3 - \ell_1 - \ell_2 + r_1^3 + r_2^3 - r_1 - r_2 \right) \\
&=
C(\ell,r) + \ell_1\ell_2 r + r_1r_2 \ell,
\end{align*}
as claimed.
\end{proof}

Now let's return to the planted partition model.
\begin{lemma}
Suppose graph $G = (V = [n], E)$ is generated from the $(n,p,q)$-planted partition model with $p > q$. Let $T^*$ be a tree whose topmost split divides $V$ into the planted clusters $L,R$. Meanwhile, let $T$ be a tree that is not $\epsilon$-good, for some $0 < \epsilon < 1/6$. Then
$$ \E[\cost_G(T)] > \E[\cost_G(T^*)] + \frac{1}{16} \epsilon (p-q) n^3 .$$
\label{lemma:imperfect}
\end{lemma}
\begin{proof}
Invoking equation~(\ref{eq:expected-cost}), we have
\begin{align*}
\E[\cost_G(T)] - \E[\cost_G(T^*)]
&=
(p-q) \left( \cost_{H(n/2,n/2)}(T) - \cost_{H(n/2,n/2)}(T^*) \right) \\
&=
(p-q) \left( \cost_{H(n/2,n/2)}(T) - C(n/2,n/2) \right) . 
\end{align*}
Thus our goal is to lower-bound the excess cost of $T$, above the optimal value of $C(n/2,n/2)$, due to the non-laminarity of its splits. Lemma~\ref{lemma:two-cliques} gives us the excess cost due to any particular split in $T$, and we will, in general, need to sum this over several splits.

Consider the uppermost node in $T$ whose leaves include at least a $(1-\epsilon)$ fraction of $L$ as well as of $R$. Call these points $L_o \subseteq L$ and $R_o \subseteq R$, and denote the split at this node by
$$ (L_o \cup R_o) \rightarrow (L_1 \cup R_1, L_2 \cup R_2), $$
where $L_1, L_2$ is a partition of $L_o$ and $R_1, R_2$ is a partition of $R_o$. It helps to think of $L_o$ has having been divided into a smaller part and a larger part; and likewise with $R_o$. Since $T$ is not $\epsilon$-good, it must either be the case that one of the smaller parts is not small enough, or that the two smaller parts have landed up on the same side of the split instead of on opposite sides. Formally, one of the two following situations must hold:
\begin{enumerate}
\item $\min(|L_1|, |L_2|) > \epsilon n/2$ or $\min(|R_1|, |R_2|) > \epsilon n/2$
\item ($|L_1| \leq \epsilon n/2$ and $|R_1| \leq \epsilon n/2$) or ($|L_2| \leq \epsilon n/2$ and $|R_2| \leq \epsilon n/2$)
\end{enumerate}
In the first case, assume without loss of generality that $|L_1|, |L_2| > \epsilon n/2$. Then the excess cost due to the non-laminarity of the split is, from Lemma~\ref{lemma:two-cliques}, at least
$$ |L_1| \cdot |L_2| \cdot |R| > \frac{\epsilon n}{2} \cdot \frac{(1-2\epsilon) n}{2} \cdot \frac{(1-\epsilon) n}{2} $$
by construction.

In the second case, suppose without loss of generality that $|L_1|, |R_1| \leq \epsilon n/2$. Since we have chosen the uppermost node whose leaves include $(1-\epsilon)$ of both $L$ and $R$, we may also assume $|L_2| < (1-\epsilon)n/2$. The excess cost of this non-laminar split and of all other splits $S \rightarrow (S_1, S_2)$ on the path up to the root is, again by Lemma~\ref{lemma:two-cliques}, at least
\begin{align*}
\sum_{S \rightarrow (S_1, S_2)} |L \cap S_1| \cdot |L \cap S_2| \cdot |R \cap S|
&\geq
\sum_{S \rightarrow (S_1, S_2)} \min(|L \cap S_1|, |L \cap S_2|) \cdot \frac{(1-2\epsilon)n}{2} \cdot \frac{(1-\epsilon)n}{2} \\
&>
\frac{\epsilon n}{2} \cdot \frac{(1-2\epsilon)n}{2} \cdot \frac{(1-\epsilon)n}{2},
\end{align*}
where the last inequality comes from observing that the smaller portions of each split of $L$ down from the root have a combined size of exactly $|L \setminus L_2| > \epsilon n/2$.

In either case, we have
$$
\cost_{H(n/2,n/2)}(T) 
\ > \ C(n/2,n/2) + \epsilon(1-\epsilon)(1-2\epsilon)\frac{n^3}{8}
\ \geq \ C(n/2,n/2) + \frac{\epsilon n^3}{16},
$$
which we substitute into our earlier characterization of expected cost.
\end{proof}

Thus there is a significant difference in expected cost between an optimal tree and one that is not $\epsilon$-good for large enough $\epsilon$. We now study the discrepancy between $\cost(T)$ and its expectation.
\begin{lemma}
Let $G = (V = [n], E)$ be a random draw from the $(n,p,q)$-planted partition model. Then for any $0 < \delta < 1$, with probability at least $1-\delta$,
$$
\max_{T} \left| \cost_G(T) - \E[\cost_G(T)] \right|
\ \leq \
\frac{n^2}{2} \sqrt{2n \ln 2n + \ln \frac{2}{\delta}},
$$
where the maximum is taken over all trees with leaves $V$.
\label{lemma:large-deviation}
\end{lemma}
\begin{proof}
Fix any tree $T$. Then $\cost_G(T)$ is a function of ${n \choose 2}$ independent random variables (the presence or absence of each edge); moreover, any one variable can change it by at most $n$. Therefore, by McDiarmid's inequality~\cite{M89}, for any $t > 0$,
$$ \pr(|\cost_G(T) - \E \cost_G(T)| > t n^2)
\ \leq \
2 \exp\left( - \frac{2t^2n^4}{{n \choose 2} n^2}  \right)
\ \leq \
2e^{-4t^2} .
$$
Next, we need an upper bound on the number of possible hierarchical clusterings on $n$ points. Imagine reconstructing any given tree $T$ in a bottom-up manner, in a sequence of $n-1$ steps, each of which involves merging two existing nodes and assigning them a parent. The initial set of points can be numbered $1,\ldots, n$ and new internal nodes can be assigned numbers $n+1,\ldots, 2n-1$ as they arise. Each merger chooses from at most $(2n)^2$ possibilities, so there are at most $(2n)^{2n}$ possible trees overall.

Taking a union bound over these trees, with $4t^2 = \ln (2(2n)^{2n} / \delta)$, yields the result.
\end{proof}

Putting these pieces together gives a characterization of the tree that minimizes the cost function, for data from a simple planted partition.
\begin{thm}
Suppose $G = (V = [n], E)$ is generated at random from an $(n,p,q)$-planted partition model. Let $T$ be the tree that minimizes $\cost_G(\cdot)$.  Then for any $\delta > 0$, with probability at least $1-\delta$, $T$ is $\epsilon$-good for
$$ \epsilon = \frac{16}{p-q} \sqrt{\frac{2\ln 2n}{n} + \frac{1}{n^2} \ln \frac{2}{\delta}} ,$$
provided this quantity is at most $1/6$.
\end{thm}
\begin{proof}
Let $T^*$ be a tree whose top split divides $V$ into its true clusters, and let $T$ be the minimizer of $\cost_G(\cdot)$. Since $\cost_G(T) \leq \cost_G(T^*)$, it follows from Lemma~\ref{lemma:large-deviation} that
$$
\E [\cost_G(T)] \ \leq \ 
\E[\cost_G(T^*)] + n^2 \sqrt{2n \ln 2n + \ln \frac{2}{\delta}}
.$$
Now, if $T$ fails to be $\epsilon$-good for some $0 < \epsilon \leq 1/6$, then by Lemma~\ref{lemma:imperfect},
$$
\epsilon
\ < \ 
\frac{16}{(p-q)n^3} \left( \E [\cost_G(T)] - \E [\cost_G(T^*)] \right) .
$$
Combining these inequalities yields the theorem statement.
\end{proof}

In summary, when data is generated from a simple partition model, the tree that minimizes $\cost(T)$ almost perfectly respects the planted clusters, misplacing at most an $O(\sqrt{(\ln n)/n})$ fraction of the data.

\section{Hardness of finding the optimal clustering}
\label{sec:hardness}

In this section, we will see that finding the tree that minimizes the cost function is NP-hard. We begin by considering a related problem.

\subsection{Maximizing the cost function}

Interestingly, the problem of {\it maximizing} $\cost(T)$ is equivalent to the problem of minimizing it. To see this, let $G = (V,E)$ be any graph with unit edge weights, and let $G^c = (V, E^c)$ denote its complement. Pick any tree $T$ and assess its suitability as a hierarchical clustering of $G$ and also of $G^c$:
\begin{align*}
\cost_G(T) + \cost_{G^c}(T) 
&=
\sum_{\{i,j\} \in E} |\leaves(T[i \vee j])| + \sum_{\{i,j\} \in E^c} |\leaves(T[i \vee j])| \\
&=
\sum_{\{i,j\}} |\leaves(T[i \vee j])| \\
&= 
\cost_{K_n}(T)
\ = \ 
\frac{n^3 - n}{3}
\end{align*}
where $K_n$ is the complete graph on $n$ nodes, and the last equality is from Theorem~\ref{thm:complete}. In short: a tree that maximizes $\cost(T)$ on $G$ also minimizes $\cost(T)$ on $G^c$, and vice versa. A similar relation holds in the presence of edge weights, where the complementary weights $w^c$ are chosen so that $w(i,j) + w^c(i,j)$ is constant across all $i,j$.

In proving hardness, it will be convenient to work with the maximization version. 

\subsection{A variant of not-all-equal SAT}

We will reduce from a special case of not-all-equal SAT that we call \naesat:
\begin{quote}
\naesat

{\it Input:} A Boolean formula in CNF such that each clause has either two or three literals, and each variable appears exactly three times: once in a 3-clause and twice, with opposite polarities, in 2-clauses.

{\it Question:} Is there an assignment to the variables under which every clause has at least one satisfied literal and one unsatisfied literal?
\end{quote}

In a commonly-used version of NAESAT, which we will take to be standard, there are exactly three literals per clause but each variable can occur arbitrarily often. Given an instance $\phi(x_1, \ldots, x_n)$ of this type, we can create an instance $\phi'$ of \naesat as follows:
\begin{itemize}
\item Suppose variable $x_i$ occurs $k$ times in $\phi$. Without loss of generality, $k > 1$ (otherwise we can discard the containing clause with impunity). Create $k$ new variables $x_{i1}, \ldots, x_{ik}$ to replace these occurrences.
\item Enforce agreement amongst $x_{i1}, \ldots, x_{ik}$ by adding implications
$$ (\overline{x}_{i1} \vee x_{i2}), \ (\overline{x}_{i2} \vee x_{i3}), \ldots, (\overline{x}_{ik} \vee x_{i1}) .$$
These $k$ clauses are satisfied $\Longleftrightarrow$ they are not-all-equal satisfied $\Longleftrightarrow$ $x_{i1} = x_{i2} = \cdots = x_{ik}$.
\end{itemize}
Then $\phi'$ has the required form, and moreover is not-all-equal satisfiable if and only if $\phi$ is not-all-equal satisfiable.

\subsection{Hardness of maximizing $\cost(T)$}

We now show that it is NP-hard to find a tree that maximizes $\cost(T)$.
\begin{thm}
Given an instance $\phi$ of \naesat, we can in polynomial time specify a weighted graph $G = (V,E,w)$ and integer $M$ such that
\begin{itemize}
\item $\max_{e \in E} w(e)$ is polynomial in the size of $\phi$, and
\item $\phi$ is not-all-equal satisfiable if and only if there exists a tree $T$ with $\cost_G(T) \geq M$.
\end{itemize}
\label{thm:hardness}
\end{thm}
\begin{proof}
Suppose instance $\phi(x_1, \ldots, x_n)$ has $m$ clauses of size three and $m'$ clauses of size two. We will assume that certain redundancies are removed from $\phi$:
\begin{itemize}
\item If there exists a 2-clause $C$ whose literals also appear in a 3-clause $C'$, then $C'$ is redundant and can be removed. 
\item The same holds if the 3-clause $C'$ contains the two literals of $C$ with polarity reversed. This is because $(x \vee y)$ is not-all-equal satisfied if and only if $(\overline{x} \vee \overline{y})$ is not-all-equal satisfied.
\item Likewise, if there is a 2-clause $C$ whose literals also appear in a 2-clause $C'$, but with reversed polarity, then $C'$ can be removed.
\item Any clause that contains both a literal and its negation can be removed.
\end{itemize}

Based on $\phi$, construct a weighted graph $G = (V,E,w)$ with $2n$ vertices, one per literal ($x_i$ and $\overline{x}_i$). The edges $E$ fall into three categories:
\begin{enumerate}
\item For each 3-clause, add six edges to $E$: one edge joining each pair of literals, and one edge joining the negations of those literals. These form two triangles. For instance, the clause $(x_1 \vee \overline{x}_2 \vee x_3)$ becomes:

\begin{center}
\includegraphics[width=2.75in]{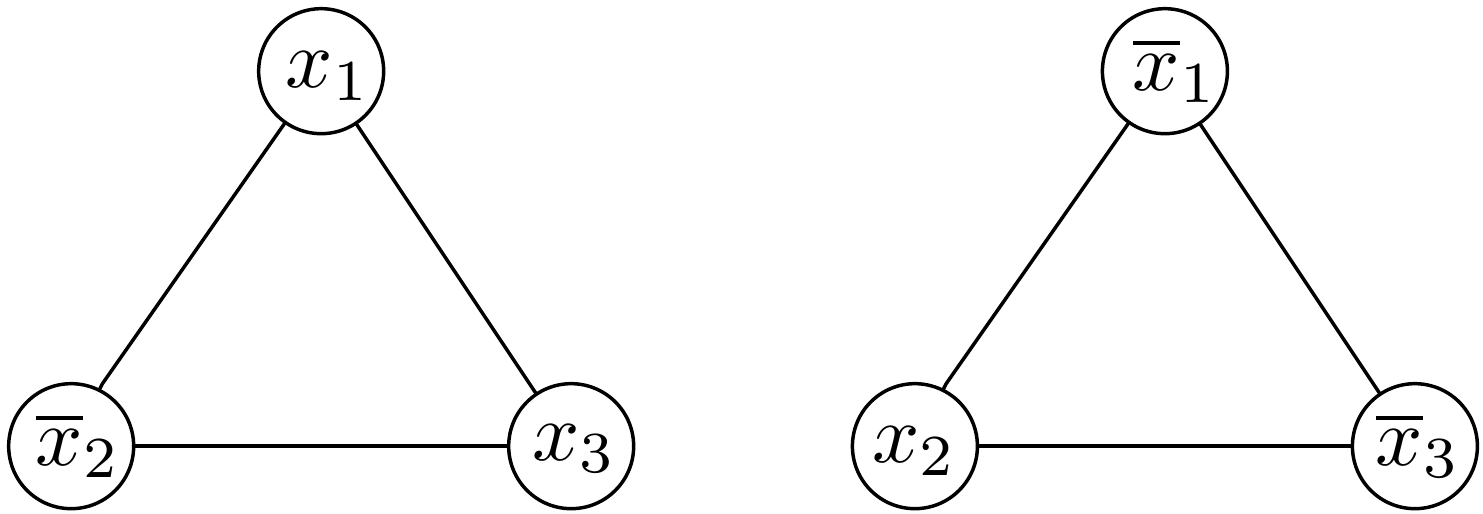}
\end{center}

These edges have weight 1.

\item Similarly, for each 2-clause, add two edges of unit weight: one joining the two literals, and one joining their negations.

\item Finally, add $n$ edges $\{x_i, \overline{x}_i\}$ of weight $W = 2nm+1$.
\end{enumerate}
Under the definition of \naesat, and with the removal of redundancies, these $6m + 2m' + n$ edges are all distinct.

Now suppose $\phi$ is not-all-equal satisfiable. Let $V^+ \subset V$ be the positive literals under the satisfying assignment, and $V^-$ the negative literals. Consider a two-level tree $T$:

\begin{center}
\includegraphics[width=2in]{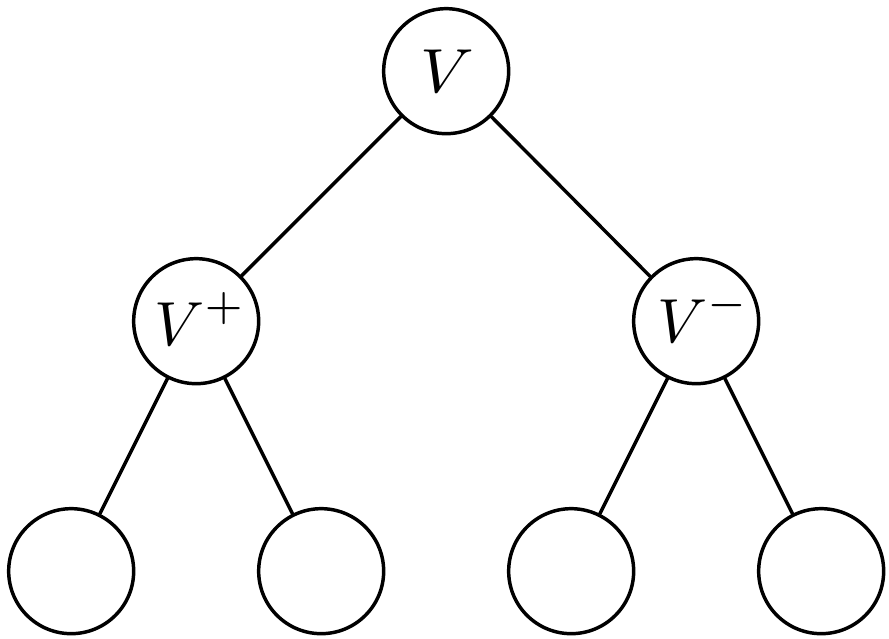}
\end{center}

The top split cuts all edges of $G$, except for one edge per triangle (and there are $2m$ of these). The cost of this split is
$$ |V| \cdot w(V^+, V^-) = 2n(4m + 2m' + nW) .$$
By the symmetry of the construction, the remaining $2m$ edges are evenly split between $V^+$ and $V^-$. The $m$ edges that lie entirely within $V^+$ contain at most one instance of each variable. Therefore, they are node-disjoint, and a single split suffices to cut all of them. The same holds for $V^-$. The total cost of the second level of splits is thus $|V^+| \cdot m + |V^-| \cdot m = 2nm$. In sum,
$$ \cost_G(T) = 2n(4m + 2m' + nW) + 2nm = 10nm + 4 nm' + 2n^2W .$$
Call this $M$.

Conversely, suppose there exists a tree $T$ of cost $\geq M$. The top split of $T$ must cut all of the edges $\{x_i, \overline{x}_i\}$; if not, $\cost(T)$ would be at most $2n(6m + 2m' + nW) - W < M$. Therefore, the split divides $V$ into two sets of size $n$, call them $V^+$ and $V^-$, such that each variable appears with one polarity in $V^+$ and the opposite polarity in $V^-$. 

This top split $V \rightarrow (V^+, V^-)$ necessarily leaves at least one edge per triangle untouched. However, it must cut all other edges, otherwise the total cost would again fall below $M$. It follows that $V^+$ is a not-all-equal satisfying assignment for $\phi$.
\end{proof}

\section{A greedy approximation algorithm}
\label{sec:alg}

Given a graph $G$ with weighted edges, we have seen that it is NP-hard to find a tree that minimizes $\cost(T)$. We now consider top-down heuristics that begin by choosing a split $V \rightarrow (S, V \setminus S)$ according to some criterion, and then recurse on each half. What a suitable split criterion?

The cost of a split $(S,V \setminus S)$ is $|V| \cdot w(S, V\setminus S)$. Ideally, we would like to shrink the node set substantially in the process, since this reduces the multiplier on subsequent splits. For a node chosen at random, the expected amount by which its cluster shrinks as a result of the split is
$$ \frac{|S|}{|V|} \cdot |V \setminus S| + \frac{|V\setminus S|}{|V|} \cdot |S| = \frac{2 |S| |V \setminus S|}{|V|} .$$ 
A natural greedy criterion would therefore be to choose the split that yields the maximum shrinkage per unit cost, or equivalently, the minimum ratio
$$ \frac{w(S,V \setminus S)}{|S| \cdot |V\setminus S|} .$$
This is known as the {\it sparsest cut} and has been studied intensively in a wide range of contexts. Although it is NP-hard to find an optimal cut, a variety of good approximation algorithms have been developed~\cite{KL70,LR99,ARV09,vL07}. We will assume simply that we have a heuristic whose approximation ratio, on graphs of $n$ nodes, is at most $\alpha_n$ times optimal, for some positive nondecreasing sequence $(\alpha_n)$. For instance, the Leighton-Rao algorithm~\cite{LR99} has $\alpha_n = O(\log n)$.

\begin{figure}
\begin{center}
\framebox[6.25in]{
\begin{minipage}[t]{6in}
\begin{tt}
\begin{tabbing}
\underline{function MakeTree(V)} \\
\= If $|V| = 1$: return leaf containing the singleton element in $V$ \\
      \> Let $(S,V\setminus S)$ be an $\alpha_n$-approximation to the sparsest cut of $V$ \\
      \> $\mbox{LeftTree} = \mbox{MakeTree}(S)$ \\
      \> $\mbox{RightTree} = \mbox{MakeTree}(V \setminus S)$ \\
      \> Return $[\mbox{LeftTree}, \mbox{RightTree}]$
\end{tabbing}
\end{tt}
\end{minipage}}
\end{center}
\caption{A top-down heuristic for finding a hierarchical clustering that approximately minimizes $\cost(T)$.}
\label{fig:alg}
\end{figure}

The resulting hierarchical clustering algorithm is shown in Figure~\ref{fig:alg}. We will now see that it returns a tree of cost at most $O(\alpha_n \log n)$ times optimal. The first step is to show that if there exists a low cost tree, there must be a correspondingly sparse cut of $V$.
\begin{lemma}
Pick any tree $T$ on $V$. There exists a partition $A,B$ of $V$ such that
$$ \frac{w(A,B)}{|A| \cdot|B|} < \frac{27}{4 |V|^3} \, \cost(T) ,$$
and for which $|V|/3 \leq |A|, |B| \leq 2|V|/3$.
\label{lemma:sparsecut}
\end{lemma}
\begin{proof}
The diagram below shows a portion of $T$, drawn in a particular way. Each node $u$ is labeled with the leaf set of its induced subtree $T[u]$. We always depict the smaller half of each cut $(A_i, B_i)$ on the right, so that $|A_i| \geq |B_i|$.

\begin{center}
\includegraphics[width=3in]{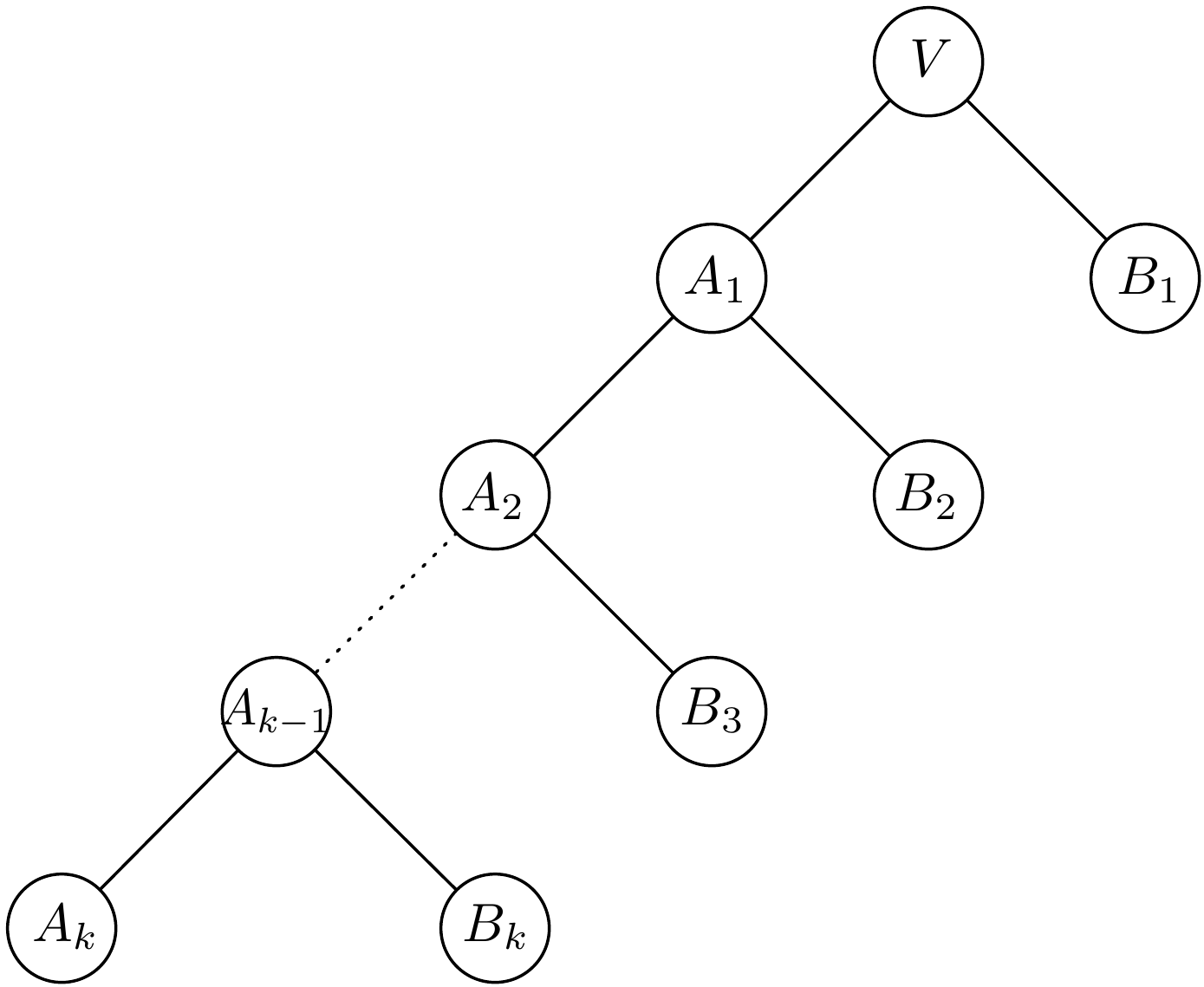}
\end{center}

Let $n = |V|$. Pick the smallest $k$ so that $|B_1| + \cdots + |B_k| \geq n/3$. This also means $|A_{k-1}| > 2n/3$ (to deal with cases when $k=1$, let $A_o = V$) and $|A_k| > n/3$. Define $A = A_k$ and $B = B_1 \cup \cdots \cup B_k$.

Now,
\begin{align*}
\mbox{cost}(T)
&\geq
w(A_1,B_1) \cdot n + w(A_2, B_2) \cdot |A_1| + \cdots + w(A_k,B_k) \cdot |A_{k-1}| \\
&> 
\frac{2n}{3} \left( w(A_1, B_1) + \cdots + w(A_k, B_k) \right) \\
&\geq
\frac{2n}{3} \, w(A, B)
\end{align*}
since the removal of all the edges in the cuts $(A_1, B_1), \ldots, (A_k, B_k)$ disconnects $A_k$ from $B_1 \cup \cdots \cup B_k$. Therefore,
$$ \frac{w(A,B)}{|A| \cdot |B|} 
\ < \ 
\frac{3}{2n} \cdot \frac{\cost(T)}{(2n/3) \cdot (n/3)}
\ = \ 
\frac{27}{4n^3} \, \cost(T),
$$
as claimed.
\end{proof}

\begin{thm}
Pick any graph $G$ on $n$ vertices, with positive edge weights $w: E \rightarrow \R^+$. Let tree $T^*$ be a minimizer of $\cost_G(\cdot)$ and let $T$ be the tree returned by the top-down algorithm of Figure~\ref{fig:alg}. Then
$$ \cost_G(T) \ \leq \ ( c_n \ln n ) \cost_G(T^*) ,$$
for $c_n = 27 \alpha_n / 4$.
\label{thm:alg}
\end{thm}
\begin{proof}
We'll use induction on $n$. The case $n=1$ is trivial since any tree with one node has zero cost.

Assume the statement holds for graphs of up to $n-1$ nodes. Now pick a weighted graph $G = (V,E,w)$ with $n$ nodes, and let $T^*$ be an optimal tree for it. By Lemma~\ref{lemma:sparsecut}, the sparsest cut of $G$ has ratio at most
$$ \frac{27}{4n^3} \cost(T^*) .$$
The top-down algorithm identifies a partition $(A,B)$ of $V$ that is an $\alpha_n$-approximation to this cut, so that
$$ \frac{w(A,B)}{|A| \cdot |B|} 
\ \leq \ \frac{27 \alpha_n}{4n^3} \, \cost(T^*) 
\  = \ \frac{c_n}{n^3} \cost(T^*),$$
and then recurses on $A$ and $B$.

We now obtain an upper bound on the cost of the best tree for $A$ (and likewise $B$). Start with $T^*$, restrict all cuts to edges within $A$, and disregard subtrees that are disjoint from $A$. The resulting tree, call it $T_A^*$, has the same overall structure as $T^*$. Construct $T_B^*$ similarly. Now, for any split $S \rightarrow (S_1, S_2)$ in $T^*$ there are corresponding (possibly empty) splits in $T_A^*$ and $T_B^*$. Moreover, the cut edges in $T_A^*$ and $T_B^*$ are disjoint subsets of those in $T^*$; hence the cost of this particular split in $T_A^*$ and $T_B^*$ combined is at most the cost in $T^*$. Formally,
\begin{align*}
(\mbox{cost of split in $T_A^*$}) + (\mbox{cost of split in $T_B^*$})
&=
|S \cap A| \cdot w(S_1 \cap A, S_2 \cap A) + |S \cap B| \cdot w(S_1 \cap B, S_2 \cap B) \\
&\leq
|S| \cdot w(S_1 \cap A, S_2 \cap A) + |S| \cdot w(S_1 \cap B, S_2 \cap B) \\
&\leq
|S| \cdot w(S_1, S_2)
\ = \ 
\mbox{cost of split in $T^*$}.
\end{align*}
Summing over all splits, we have
$$ \cost(T_A^*) + \cost(T_B^*) \leq \cost(T^*) .$$
Here the costs on the left-hand side are with respect to the subgraphs of $G$ induced by $A$ and $B$, respectively.

Without loss of generality, $|A| = pn$ and $|B| = (1-p) n$, for some $0 < p < 1/2$. Recall that our algorithm recursively constructs trees, say $T_A$ and $T_B$, for subsets $A$ and $B$. Applying the inductive hypothesis to these trees, we have
\begin{align*}
\cost(T_A) &\leq c_n \cost(T_A^*) \ln pn \\
\cost(T_B) &\leq c_n \cost(T_B^*) \ln (1-p)n
\end{align*}
where we have used the monotonicity of $(\alpha_n)$, and hence
\begin{align*}
\cost(T)
&\leq
n \cdot w(A,B) + \cost(T_A) + \cost(T_B) \\
&\leq
n \cdot |A| \cdot |B| \cdot \frac{c_n}{n^3} \cdot \cost(T^*) + c_n \cost(T_A^*) \ln pn + c_n \cost(T_B^*) \ln (1-p)n \\
&\leq
c_n p (1-p) \cdot \cost(T^*) + c_n \cost(T_A^*) \ln (1-p)n + c_n \cost(T_B^*) \ln (1-p)n \\
&\leq
c_n p \cdot \cost(T^*) + c_n \cost(T^*) \ln (1-p) n \\
&=
c_n \cost(T^*) \left( p + \ln (1-p) + \ln n \right) \\
&\leq 
c_n \cost(T^*) \ln n,
\end{align*}
as claimed.
\end{proof}

\section{A generalization of the cost function}

A more general objective function for hierarchical clustering is
$$ \cost_G(T) = \sum_{\{i,j\} \in E} w_{ij} \, f(|\leaves(T[i \vee j])|) ,$$
where $f$ is defined on the nonnegative reals, is strictly increasing, and has $f(0) = 0$. For instance, we could take $f(x) = \ln (1 + x)$ or $f(x) = x^2$.

Under this generalized cost function, all the properties of Section~\ref{sec:basics} continue to hold, substituting $|S|$ by $f(|S|)$ as needed, for $S \subseteq V$. However, it is no longer the case that for the complete graph, all trees have equal cost. When $G$ is the clique on four nodes, for instance, the two trees shown below ($T_1$ and $T_2$) need not have the same cost. 

\begin{center}
\includegraphics[width=5in]{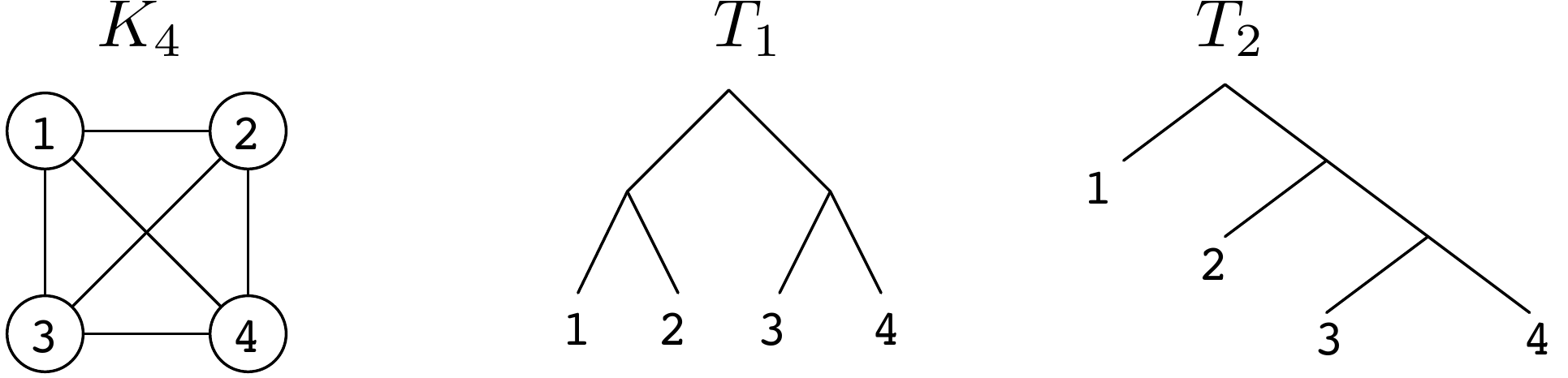}
\end{center}

The tree costs, for any $f$, are:
\begin{align*}
\cost(T_1) &= 4 f(4) + 2 f(2) \\
\cost(T_2) &= 3 f(4) + 2 f(3) + f(2) 
\end{align*}
Thus $T_1$ is preferable if $f$ is concave (in which case $f(2) + f(4) \leq 2 f(3)$), while the tables are turned if $f$ is convex.

Nevertheless, the greedy top-down heuristic continues to yield a provably good approximation. It needs to be modified slightly: the split chosen is (an $\alpha_n$-approximation to) the minimizer $S \subseteq V$ of
$$
\frac{w(S, V \setminus S)}{\min(f(|S|), f(|V \setminus S|))} 
\mbox{\ \ subject to\ \ }\frac{1}{3}|V| \leq |S| \leq \frac{2}{3}|V| .
$$

Lemma~\ref{lemma:sparsecut} and Theorem~\ref{thm:alg} need slight revision.
\begin{lemma}
Pick any tree $T$ on $V$. There exists a partition $A,B$ of $V$ such that
$$ \frac{w(A,B)}{\min(f(|A|),f(|B|))} \leq \frac{\cost(T)}{f(\lfloor 2n/3\rfloor) f(\lceil n/3 \rceil)} .$$
and for which $|V|/3 \leq |A|, |B| \leq 2|V|/3$.
\label{lemma:sparsecut-generalized}
\end{lemma}
\begin{proof}
Sets $A,B$ are constructed exactly as in the proof of Lemma~\ref{lemma:sparsecut}. Then, referring to the diagram in that proof,
\begin{align*}
\mbox{cost}(T)
&\geq
w(A_1,B_1) \cdot f(n) + w(A_2, B_2) \cdot f(|A_1|) + \cdots + w(A_k,B_k) \cdot f(|A_{k-1}|) \\
&\geq 
f( \lfloor 2n/3 \rfloor) \left( w(A_1, B_1) + \cdots + w(A_k, B_k) \right) \ \geq \ f(\lfloor 2n/3 \rfloor) \, w(A, B),
\end{align*}
whereupon
$$ \frac{w(A,B)}{\min(f(|A|),f(|B|))} 
\ \leq \ 
\frac{\cost(T)}{f(\lfloor 2n/3 \rfloor) f(\lceil n/3 \rceil)}.
$$
\end{proof}

\begin{thm}
Pick any graph $G$ on $n$ vertices, with positive edge weights $w: E \rightarrow \R^+$. Let tree $T^*$ be a minimizer of $\cost(\cdot)$ and let $T$ be the tree returned by the modified top-down algorithm. Then
$$ \cost_G(T) \leq (c_n \ln n)\, \cost_G(T^*) ,$$
where
$$ c_n = 3 \alpha_n \cdot \max_{1 \leq n' \leq n} \frac{f(n')}{f(\lceil n'/3 \rceil)} .$$
\label{thm:alg-generalized}
\end{thm}
\begin{proof}
The proof outline is as before.  We assume the statement holds for graphs of up to $n-1$ nodes. Now pick a weighted graph $G = (V,E,w)$ with $n$ nodes, and let $T^*$ be an optimal tree for it. 

Let $T$ denote the tree returned by the top-down procedure. By Lemma~\ref{lemma:sparsecut-generalized}, the top split $V \rightarrow (A,B)$ satisfies $n/3 \leq |A|, |B| \leq 2n/3$ and 
$$ w(A,B) 
\ \leq \ 
\alpha_n \, {\min(f(|A|), f(|B|))} \, \frac{\cost(T^*)}{f(\lfloor 2n/3 \rfloor) f(\lceil n/3 \rceil)}
\ \leq \ 
\frac{\alpha_n}{f(\lceil n/3 \rceil)} \cost(T^*), $$
where the last inequality uses the monotonicity of $f$.

As before, the optimal tree $T^*$ can be used to construct trees $T_A^*$ and $T_B^*$ for point sets $A$ and $B$, respectively, such that $\cost(T_A^*) + \cost(T_B^*) \leq \cost(T^*)$. Meanwhile, the top-down algorithm finds trees $T_A$ and $T_B$ which, by the inductive hypothesis, satisfy
\begin{align*}
\cost(T_A) &\leq (c_n \ln |A|) \cost(T_A^*)  \\
\cost(T_B) &\leq (c_n \ln |B|) \cost(T_B^*) 
\end{align*}
Let's say, without loss of generality, that $|A| \leq |B|$. Then
\begin{align*}
\cost(T)
&\leq
f(n) \cdot w(A,B) + \cost(T_A) + \cost(T_B) \\
&\leq
f(n) \cdot w(A,B) + (c_n \ln |B|) (\cost(T_A^*) + \cost(T_B^*)) \\
&\leq
\frac{\alpha_n f(n)}{f(\lceil n/3 \rceil)} \cost(T^*) + (c_n \ln n + c_n \ln (2/3)) \cost(T^*) \\
&\leq
\frac{c_n}{3} \cost(T^*) + (c_n \ln n  - c_n/3) \cost(T^*)
\ = \ 
(c_n \ln n) \cost(T^*),
\end{align*}
as claimed.
\end{proof}

\subsection*{Acknowledgements}

The author is grateful to the National Science Foundation for support under grant IIS-1162581.

\bibliography{/Users/dasgupta314/Dropbox/PAPERS/sanjoy}

\begin{thebibliography}{10}

\bibitem{ARV09}
S.~Arora, S.~Rao, and U.V. Vazirani.
\newblock Expander flows, geometric embeddings and graph partitioning.
\newblock {\em Journal of the Association for Computing Machinery}, 56(2),
  2009.

\bibitem{B85a}
R.B. Boppana.
\newblock Eigenvalues and graph bisection: an average-case analysis.
\newblock In {\em IEEE Symposium on Foundations of Computer Science}, pages
  280--285, 1985.

\bibitem{CDKL14}
K.~Chaudhuri, S.~Dasgupta, S.~Kpotufe, and U.~von Luxburg.
\newblock Consistent procedures for cluster tree estimation and pruning.
\newblock {\em IEEE Transactions on Information Theory}, 60(12):7900--7912,
  2014.

\bibitem{DL05}
S.~Dasgupta and P.M. Long.
\newblock Performance guarantees for hierarchical clustering.
\newblock {\em Journal of Computer and System Sciences}, 70(4):555--569, 2005.

\bibitem{EBW15}
J.~Eldridge, M.~Belkin, and Y.~Wang.
\newblock Beyond {H}artigan consistency: merge distortion metric for
  hierarchical clustering.
\newblock In {\em 28th Annual Conference on Learning Theory}, 2015.

\bibitem{F04}
J.~Felsenstein.
\newblock {\em Inferring Phylogenies}.
\newblock Sinauer, 2004.

\bibitem{H85}
J.A. Hartigan.
\newblock Statistical theory in clustering.
\newblock {\em Journal of Classification}, 2:63--76, 1985.

\bibitem{HTF09}
T.~Hastie, R.~Tibshirani, and J.~Friedman.
\newblock {\em The Elements of Statistical Learning}.
\newblock Springer, 2nd edition, 2009.

\bibitem{JS71}
N.~Jardine and R.~Sibson.
\newblock {\em Mathematical Taxonomy}.
\newblock John Wiley, 1971.

\bibitem{KL70}
B.W. Kernighan and S.~Lin.
\newblock An efficient heuristic procedure for partitioning graphs.
\newblock {\em Bell System Technical Journal}, 49(2):291--–307, 1970.

\bibitem{LR99}
F.T. Leighton and S.~Rao.
\newblock Multicommodity max-flow min-cut theorems and their use in designing
  approximation algorithms.
\newblock {\em Journal of the Association for Computing Machinery},
  46(6):787--832, 1999.

\bibitem{LNRW10}
G.~Lin, C.~Nagarajan, R.~Rajaraman, and D.P. Williamson.
\newblock A general approach for incremental approximation and hierarchical
  clustering.
\newblock {\em SIAM Journal on Computing}, 39:3633--3669, 2010.

\bibitem{M89}
C.~Mc{D}iarmid.
\newblock On the method of bounded differences.
\newblock {\em Surveys in Combinatorics}, 141:148–--188, 1989.

\bibitem{M01}
F.~McSherry.
\newblock Spectral partitioning of random graphs.
\newblock In {\em IEEE Symposium on Foundations of Computer Science}, pages
  529--537, 2001.

\bibitem{N03}
R.M. Neal.
\newblock Density modeling and clustering using {D}irichlet diffusion trees.
\newblock In J.M. Bernardo et~al., editors, {\em Bayesian Statistics 7}, pages
  619--629. Oxford University Press, 2003.

\bibitem{P06}
C.G. Plaxton.
\newblock Approximation algorithms for hierarchical location problems.
\newblock {\em Journal of Computer and System Sciences}, 72:425--443, 2006.

\bibitem{SS63}
R.R. Sokal and P.H.A. Sneath.
\newblock {\em Numerical Taxonomy}.
\newblock W.H. Freeman, 1963.

\bibitem{vL07}
U.~von Luxburg.
\newblock A tutorial on spectral clustering.
\newblock {\em Statistics and Computing}, 17(4):395--416, 2007.

\end{thebibliography}
\bibliographystyle{plain}

\end{document}